\ifx\JVER\undefined%
\newcommand{\InJVer}[1]{}%
\newcommand{\InNotJVer}[1]{#1}%
\else \newcommand{\InJVer}[1]{#1}%
\newcommand{\InNotJVer}[1]{}%
\fi

\InNotJVer{%
   \documentclass[12pt]{article}%
   \usepackage[cm]{fullpage}%
}%
\InJVer{%
   \documentclass[format=acmsmall, review=false, screen=true]{acmart} }

\usepackage{graphicx}%
\usepackage{amsmath}%
\usepackage{color}%
\usepackage{xspace}%
\usepackage{paralist}
\usepackage{mleftright}%
\usepackage{amssymb}%
\usepackage{xcolor}%
\InNotJVer{%
   \usepackage{picins}%
}
\usepackage{caption}%

\usepackage{stmaryrd}%

\usepackage{abstract}

\InNotJVer{%
   \usepackage[amsmath,thmmarks]{ntheorem}%
   \theoremseparator{.}%

   \usepackage{titlesec}%
   \titlelabel{\thetitle. }%

   \setcounter{secnumdepth}{4}
}

\usepackage{hyperref}%
\hypersetup{%
   breaklinks,%
   colorlinks=true,%
   linkcolor=[rgb]{0.45,0.0,0.0},%
   citecolor=[rgb]{0,0,0.45}
}
\newcommand{\hrefb}[3][black]{\href{#2}{\color{#1}{#3}}}%

\numberwithin{figure}{section}%
\numberwithin{table}{section}%
\numberwithin{equation}{section}%

\renewcommand{\Re}{\mathbb{R}}%
\theoremstyle{plain}%

\InNotJVer{
   \newtheorem{theorem}{Theorem}[section]
   \newtheorem{lemma}[theorem]{Lemma}%

}
\newtheorem{theorem:unnumbered}{Theorem}[section] 

\newtheorem{claim}[theorem]{Claim}%
\InNotJVer{%
   \newtheorem*{example:unnumbered}[theorem:unnumbered]{Example}%
}
\InJVer{%
   \newtheorem{example:unnumbered}[theorem]{Example}
}

\newtheorem{problem}[theorem]{Problem}%
\newtheorem{observation}[theorem]{Observation}

\InNotJVer{%
   \theoremstyle{plain}%
   \theoremheaderfont{\sf}%
   \theorembodyfont{\upshape}%
   \newtheorem*{remark:unnumbered}[theorem]{Remark}%
   \newtheorem{remark}[theorem]{Remark}%
}
\InJVer{%
   \newtheorem{remark}[theorem]{Remark}%
}
\newtheorem{defn}[theorem]{Definition}
\InJVer{%
      \newtheorem{remark:unnumbered}[theorem]{Remark}%
}

\newcommand{\HLinkShort}[2]{\hyperref[#2]{#1\ref*{#2}}}
\newcommand{\HLink}[2]{\hyperref[#2]{#1~\ref*{#2}}}
\newcommand{\HLinkPage}[2]{\hyperref[#2]{#1~\ref*{#2}%
      $_\text{p\pageref{#2}}$}}
\newcommand{\HLinkPageOnly}[1]{\hyperref[#1]{Page~\refpage*{#1}%
      $_\text{p\pageref{#1}}$}}

\newcommand{\HLinkSuffix}[3]{\hyperref[#2]{#1\ref*{#2}{#3}}}
\newcommand{\HLinkPageSuffix}[3]{\hyperref[#2]{#1\ref*{#2}%
      #3$_\text{p\pageref{#2}}$}}

\newcommand{\deflab}[1]{\label{def:#1}}
\newcommand{\defref}[1]{\HLink{Definition}{def:#1}}%

\newcommand{\figlab}[1]{\label{fig:#1}}
\newcommand{\figref}[1]{\HLink{Figure}{fig:#1}}

\newcommand{\itemlab}[1]{\label{item:#1}}
\newcommand{\itemref}[1]{\HLinkSuffix{(}{item:#1}{)}}

\newcommand{\remlab}[1]{\label{rem:#1}}
\newcommand{\remref}[1]{\HLink{Remark}{rem:#1}}%

\providecommand{\lemlab}[1]{\label{xlemma:#1}}
\renewcommand{\lemlab}[1]{\label{xlemma:#1}}
\newcommand{\lemref}[1]{\HLink{Lemma}{xlemma:#1}}%

\newcommand{\clmlab}[1]{\label{claim:#1}}
\newcommand{\clmref}[1]{\HLink{Claim}{claim:#1}}%

\newcommand{\problab}[1]{\label{problem:#1}}
\newcommand{\probref}[1]{\HLink{Problem}{problem:#1}}%

\newcommand{\thmlab}[1]{{\label{theo:#1}}}
\newcommand{\thmref}[1]{\HLink{Theorem}{theo:#1}}

\newcommand{\seclab}[1]{\label{sec:#1}}
\newcommand{\secref}[1]{\HLink{Section}{sec:#1}}

\newcommand{\permut}[1]{\left\langle {#1} \right\rangle}
\newcommand{\DotProd}[2]{\permut{{#1},{#2}}}

\newcommand{\nnY}[2]{\Math{\mathrm{nn}}\pth{#1, #2}}

\providecommand{\eqlab}[1]{}%
\renewcommand{\eqlab}[1]{\label{equation:#1}}
\newcommand{\Eqref}[1]{\HLinkSuffix{Eq.~(}{equation:#1}{)}}

\newcommand{\vspan}{\mathsf{span}}

\InNotJVer{%
   \newcommand{\myqedsymbol}{\rule{2mm}{2mm}}
   \theoremheaderfont{\em}%
   \theorembodyfont{\upshape}%
   \theoremstyle{nonumberplain}%
   \theoremseparator{}%
   \theoremsymbol{\myqedsymbol}%
   \newtheorem{proof}{Proof:}%
}

\newcommand{\niter}{\Math{M}}

\newcommand{\Line}{\ell}%

\newcommand{\Term}[1]{\textsf{#1}}%

\DefineNamedColor{named}{RedViolet} {cmyk}{0.07,0.90,0,0.34}

\newcommand{\Mat}{\Math{M}}%
\newcommand{\MatA}{\Math{M'}}%

\definecolor{blue25}{rgb}{0, 0, 11}%
\newcommand{\emphic}[2]{%
   \textcolor{blue25}{%
      \textbf{\emph{#1}}}%
   \index{#2}}

\newcommand{\emphi}[1]{\emphic{#1}{#1}}

\newcommand{\dOneSideH}[2]{\Math{d}\pth{#1 \rightarrow #2}}%
\newcommand{\distH}[2]{d_H\pth{#1, #2}}%

\newcommand{\diamX}[1]{\mathrm{diam}\pth{#1}}%

\newcommand{\kalg}{\Math{k_{\mathrm{alg}}}}%

\newcommand{\kopt}{\Math{k_{\mathrm{opt}}}}%
\newcommand{\koptY}[2]{\kopt\pth{#1, #2}}%
\newcommand{\koptZ}[3]{\kopt\pth{#1, #2, #3}}%

\renewcommand{\th}{th\xspace}

\newcommand{\atgen}{\symbol{'100}}
\newcommand{\SarielThanks}[1]{\thanks{Department of Computer Science;
      University of Illinois; 201 N. Goodwin Avenue; Urbana, IL,
      61801, USA; {\tt sariel\atgen{}illinois.edu}; {\tt
         \url{http://sarielhp.org/}.} #1}}
\newcommand{\BenThanks}[1]{%
   \thanks{%
      Department of Computer Science;
      University of Texas at Dallas; Richardson, TX 75080, USA; 
      {\tt benjamin.raichel\atgen{}utdallas.edu}; {\tt
         \url{http://utdallas.edu/\string~benjamin.raichel}.} #1}}

\newcommand{\ceil}[1]{\left\lceil {#1} \right\rceil}

\newcommand{\etal}{\textit{et~al.}\xspace}

\newcommand{\Caratheodory}{Carath\'eodory\xspace}

\newcommand{\Pin}{\Math{\PntSet_{\mathrm{in}}}}%
\newcommand{\Pout}{\Math{\PntSet_{\mathrm{out}}}}%

\newcommand{\cardin}[1]{\left| {#1} \right|}%

\newcommand{\VC}{\Term{VC}\xspace}%
\newcommand{\LP}{\Term{LP}\xspace}%
\newcommand{\RangeSpace}{\Math{\mathcal{S}}}%
\newcommand{\hpZ}[3]{%
   \hplane^+\pth{#1, #2, #3}%
}

\newcommand{\Set}[2]{\left\{ #1 \;\middle\vert\; #2 \right\}}
\newcommand{\pth}[1]{\mleft({#1}\mright)}%
\newcommand{\hplane}{\Math{h}}%
\newcommand{\hplaneA}{\Math{z}}%
\newcommand{\CHChar}{{\Math{\mathcal{C}}}}
\newcommand{\CHX}[1]{\CHChar_{#1}}

\newcommand{\distSet}[2]{\Math{d}\pth{#1, #2}}
\newcommand{\norm}[1]{\left\| {#1}  \right\|}
\newcommand{\distY}[2]{\norm{#1 - #2}}

\newcommand{\ApproxQB}[2]{\ensuremath{\pth{ #1, #2}}}

\newcommand{\eVec}[1]{\mathrm{e}_{#1}}

\newcommand{\sspace}{L}%
\newcommand{\dSubSpace}[2]{\ell_{#1}\pth{#2}}%
\newcommand{\dCHY}[2]{d_{#1}\pth{#2}}
\newcommand{\annY}[2]{\mathsf{nn}_{#1}\pth{#2}}

\newcommand{\brc}[1]{\left\{ {#1} \right\}}

\newcommand{\query}{\Math{q}}

\newcommand{\pnt}{\Math{{p}}}%
\newcommand{\pntA}{\Math{{t}}}%
\newcommand{\pntB}{\Math{{y}}}%
\newcommand{\pntC}{\Math{{z}}}%
\newcommand{\pntD}{\Math{{u}}}%

\newcommand{\PntSet}{\Math{P}}%
\newcommand{\PntSetA}{\Math{U}}%

\newcommand{\PntSetC}{\Math{Z}}

\newcommand{\CSet}{\Math{{C}}}%
\newcommand{\DSet}{\Math{{D}}}%

\newcommand{\opnt}{\Math{o}}%

\newcommand{\ds}{\displaystyle}%

\newcommand{\Math}[1]{{{#1}}}%

\newcommand{\OptAprxY}[2]{\Math{\mathrm{opt}}\pth{#1, #2}}%
\newcommand{\OptAprxZ}[3]{\Math{\mathrm{opt}}\pth{#1, #2, #3}}%

\newcommand{\vect}{\Math{v}}%
\newcommand{\seg}{\Math{s}}%

\newcommand{\reach}{\delta}
\newcommand{\reachVal}{8 \eps^{1/3}}
\newcommand{\reachValHalf}{4 \eps^{1/3}}
\newcommand{\diam}{\Math{\Delta}}%

\newcommand{\eps}{\varepsilon}

\newcommand{\ballD}{\mathsf{b}}

\newcommand{\ballY}[2]{\mathrm{ball}\pth{#1, #2}}%

\newcommand{\SphereChar}{\ensuremath{\Math{\mathbb{S}}}}
\newcommand{\Sphere}[1]{{\SphereChar^{(#1)}}}

\newcommand{\obs}{\Math{p}}
\newcommand{\comb}{\Math{x}}

\newcommand{\OptSet}{\Math{P_\mathrm{opt}}}%
\newcommand{\algset}{\Math{T}}

\newcommand{\npoints}{n}

\newcommand{\idist}{\tau}%
\newcommand{\idistA}{\tau'}%

\newcommand{\remove}[1]{}%

\IfFileExists{sariel_computer.sty}{\def\sarielComp{1}}{}
\ifx\sarielComp\undefined%
\newcommand{\SarielComp}[1]{}
\newcommand{\NotSarielComp}[1]{#1}%
 \else
\newcommand{\SarielComp}[1]{#1}%
\newcommand{\NotSarielComp}[1]{}%
\fi

\SarielComp{%
   \ifx\colorMath\undefined%
   \else
       \renewcommand{\Math}[1]{{\textcolor{red}{#1}}}
   \fi
}

\InJVer{%
   \acmJournal{TALG}%
   \acmVolume{-1}%
   \acmNumber{-1}%
   \acmArticle{-1}%
   \acmYear{2525}%
   \acmMonth{3}%
   \copyrightyear{2525}%
}

\begin{document}

\title{Sparse Approximation via Generating Point Sets%
   \InNotJVer{%
      \thanks{A preliminary version of this paper appeared in SODA 16
         \cite{bhr-sagps-16}. The full version of the paper is also
         available on the arxiv \cite{bhr-sagps-15}.%
      }%
   }%
}

\InJVer{%
   \author{Avrim Blum}
   \affiliation{%
      \institution{Carnegie Mellon University}
      \department{Department of Computer Science}
      \city{Pittsburgh}
      \state{PA}
      \postcode{15213-3891}
      \country{USA}%
   }%
   \author{Sariel Har-Peled}%
   \affiliation{%
      \institution{University of Illinois, Urbana-Champaign}%
      \department{Department of Computer Science}%
      \streetaddress{201 N. Goodwin Avenue}
      \city{Urbana}%
      \state{IL}%
      \postcode{61801}
      \country{USA}%
   }
   \author{Benjamin Raichel}%
   \affiliation{%
      \institution{University of Texas at Dallas}
      \department{Department of Computer Science}
      \city{Richardson}
      \state{TX}
      \postcode{75080}
      \country{USA}
   }
}

\InNotJVer{%
   \author{%
      Avrim Blum%
      \thanks{$\Bigl.$%
         Department of Computer Science, Carnegie Mellon University;
         avrim@cs.cmu.edu.  Work was conducted while on sabbatical at
         the University of Illinois.  Work on this paper was partially
         supported by NSF award CCF-1415460, and CCF-1525971.}%
      \and%
      Sariel Har-Peled\SarielThanks{Work on this paper was partially
         supported by a %
         NSF AF awards %
         CCF-1421231, and 
         CCF-1217462.  
      }%
      \and%
      Benjamin Raichel\BenThanks{%
         Work on this paper was partially supported by the University
         of Illinois Graduate College Dissertation Completion
         Fellowship.  }%
   } }

\date{\today}%

\ifx\JVER\undefined%
\else

\fi


\InJVer{%
   \thanks{%
      Work by A.B. was conducted while on sabbatical at the University
      of Illinois, and was partially supported by NSF award
      CCF-1415460, and CCF-1525971. %
      Work by S.H. was partially supported by NSF AF awards
      CCF-1421231, and CCF-1217462. %
      Work bu B.R.  was partially supported by the University of
      Illinois Graduate College Dissertation Completion Fellowship.
   }%
   \thanks{A preliminary version of this paper appeared in SODA 16
      \cite{bhr-sagps-16}. The full version of the paper is also
      available on the arxiv \cite{bhr-sagps-15}. %
   }%
}%

\maketitle

\begin{abstract}
    For a set $\PntSet$ of $\npoints$ points in the unit ball
    $\ballD \subseteq \Re^d$, consider the problem of finding a small
    subset $\algset \subseteq \PntSet$ such that its convex-hull
    $\eps$-approximates the convex-hull of the original set.
    Specifically, the Hausdorff distance between the convex hull of
    $\algset$ and the convex hull of $\PntSet$ should be at most
    $\eps$.  We present an efficient algorithm to compute such an
    $\eps'$-approximation of size $\kalg$, where $\eps'$ is a function
    of $\eps$, and $\kalg$ is a function of the minimum size $\kopt$
    of such an $\eps$-approximation. Surprisingly, there is no
    dependence on the dimension $d$ in either of the
    bounds. Furthermore, every point of $\PntSet$ can be
    $\eps$-approximated by a convex-combination of points of $\algset$
    that is $O(1/\eps^2)$-sparse.

    Our result can be viewed as a method for sparse, convex
    autoencoding: approximately representing the data in a compact way
    using sparse combinations of a small subset $\algset$ of the
    original data. The new algorithm can be kernelized, and it
    preserves sparsity in the original input.
\end{abstract}

\section{Introduction}
 
\paragraph*{Sparse approximation and coresets.}

Let $\PntSet$ be a set of $\npoints$ points (observations) in the unit
ball $\ballD \subseteq \Re^d$, and let $\CHX{\PntSet}$ denote the
convex-hull of $\PntSet$. Consider the problem of finding a small
$\eps$-coreset $\algset \subseteq \PntSet$ for projection width; that
is, given any line $\Line$ in $\Re^d$, consider the projections of
$\CHX{\algset}$ and $\CHX{\PntSet}$ onto the line $\Line$ -- these are
two intervals $I_\algset \subseteq I_\PntSet$, and we require that
$I_\PntSet \subseteq (1+\eps)I_\algset$. Such coresets have size
$O\pth{1/\eps^{(d-1)/2}}$, and lead to numerous efficient
approximation algorithms in low-dimensions, see \cite{ahv-gavc-05}.
In particular, such an $\eps$-coreset guarantees that the Hausdorff
distance between $\CHX{\algset}$ and $\CHX{\PntSet}$ is at most
$\eps$.

While such coresets can have size
$\Omega( \allowbreak1/\eps^{(d-1)/2})$ in the worst case, data may
have structure allowing much smaller coresets to exist even in high
dimensional spaces.  For example, consider a dataset $\PntSet$ in
which all points are $\eps$-close to one of $k$ different lines.  Then
taking the extreme dataset points associated with each line results in
$2k$ points, such that every $\obs \in \PntSet$ is $2\eps$-close to
the convex hull of those points.  More generally, the union of any two
datasets which have good approximations of sizes $k$ and $k'$,
respectively, has one of size at most $k+k'$.  Thus, it is natural to
ask whether one can approximate the smallest such coreset, in terms of
both its size and approximation quality.

\paragraph*{The problem in matrix form.} 

Given a collection $\PntSet$ of $\npoints$ points (observations) in
the unit ball $\ballD \subseteq \Re^d$, viewed as column vectors, find
a $d \times k$ matrix $\Mat$ such that each $\obs \in \PntSet$ can be
approximately reconstructed as a {\em sparse, convex combination} of
the columns of $\Mat$.  That is, for each $\obs \in \PntSet$ there
exists a sparse non-negative vector $\comb$ whose entries sum to one
such that $\obs \approx \Mat\comb$.  This problem is trivial if we
allow $k=\npoints$: simply make each data point $\obs \in \PntSet$
into a column of $\Mat$, allowing the $i$\th data point to be
perfectly reconstructed using $\comb=e_i$, where $e_i$ is the $i$\th
vector in the standard basis.  The goal is to do so using
$k \ll \npoints$, so that $\Mat$ and the $\comb$'s can be viewed as an
(approximate) compressed representation of the $\obs$'s.

\begin{figure*}[t]%
    \centerline{%
       \begin{tabular}{|l|c|c|l|}%
         \hline
         Technique%
         &%
           $\eps' \equiv \distH{\CHX{\PntSet}}{\CHX{\algset}}$
         &%
           $\kalg\equiv \cardin{\algset}\Bigl.$ 
         &%
           Result\\
         \hline%
         \hline%
         $\eps$-nets%
         &%
           $\leq \eps$%
         &%
           $\Bigl. O( d \kopt \log \kopt )$
         &%
           \lemref{v:c:argument}%
           \qquad
           \qquad
         \\%
         \hline%
         Greedy set cover%
         &%
           $\leq (1+\delta)\eps$%
         &%
           $\Bigl. O\pth{\bigl. \pth{\kopt/ (\eps\delta)^{2}} \log n }$%
         &%
           \lemref{greedy}%
         \\%
         \hline%
         Greedy clustering\qquad\qquad%
         &%
           $\leq \reachVal + \eps$
         &%
           $\Bigl.O\pth{\kopt/\eps^{2/3}}$
         &%
           \begin{minipage}{3.8cm}
               \smallskip%
               \thmref{main}\\%
               No dependency on $d$ or $n$%
               \smallskip%
           \end{minipage}
         \\%
         \hline
       \end{tabular}%
    }
    \caption{Summary of our results: Given a set $\PntSet$ contained
       in the unit ball of $\Re^d$, such that there is a subset
       $\OptSet \subseteq \PntSet$ of size $\kopt$, and
       $\distH{\CHX{\PntSet}}{\CHX{\OptSet}} \leq \eps$, the above
       results compute an approximate set
       $\algset \subseteq \PntSet$.  Note, that any point in
       $\PntSet$ has an $O(1/\eps^2)$-sparse $(\eps+\eps')$-approximation
       using $\algset$, because of the underlying sparsity -- see
       \lemref{n:n}.  }
    \figlab{our:results}
\end{figure*}

\paragraph*{Input assumption.}
We are given a set $\PntSet$ of $n$ points in $\Re^d$ all with norm
at most one.  Suppose that there \emph{exists} a
$d \times \kopt$ matrix $\Mat$, such that
\smallskip%
\begin{compactenum}[\;\;(A)]
    \item  each column of $\Mat$ is a convex combination of the observations
    $\obs$, and
    \item each $\obs \in \PntSet$ can be $\eps$-approximately
    reconstructed as a convex combination of the columns of $\Mat$:
    that is, for each $\obs \in \PntSet$ there exists a non-negative
    vector $\comb$ whose entries sum to one such that
    $\distY{\obs}{\Mat\comb} \leq \eps$.
\end{compactenum}
\smallskip%
Stated geometrically, the assumption is that the input $\PntSet$ is
contained in the unit ball $\ballD$ (centered at the origin), and
there exists a set $\OptSet \subseteq \CHX{\PntSet}$, of size
$\kopt$, such that for any point $\obs \in \PntSet$, we have that
$\obs$ is $\eps$-close to $\CHX{\OptSet}$, where $\CHX{\OptSet}$
denotes the convex-hull of $\OptSet$.  Formally, being
\emph{$\eps$-close} means that the distance of $\obs$ to the set
$\CHX{\OptSet}$ is at most $\eps$.

\paragraph*{Our results.}
We present efficient algorithms for computing a $d \times \kalg$
matrix $\MatA$, consisting of $\kalg$ points of $\PntSet$, such that
each $\obs \in \PntSet$ can be $\eps'$-approximately reconstructed as
a sparse convex combination of the columns of $\MatA$, where $\kalg$
and $\eps'$ are not too large, see \figref{our:results} for details.
Here, \emph{sparse} means that only relatively few of the columns of
$\MatA$ would be used to represent (approximately) each point of
$data$.

Stated in geometric terms, the algorithm computes a set $\algset$ of
$\kalg$ points (these will be points from $\PntSet$) such that every
point in $\PntSet$ is $\eps'$-close to the convex hull of $\algset$
and moreover can be approximately reconstructed using a sparse convex
combination of $\algset$.

The reader may notice that sparsity is not mentioned in the assumption
about $\OptSet$ ($\equiv\Mat$) and yet appears in the conclusion about
$\algset$ ($\equiv\MatA$).  This is because convex combinations have
the property that sparsity can be achieved almost for free, at the
expense of a small amount of reconstruction error (see \lemref{n:n}).
This is to some extent the same reason that a large margin separator
can be represented using a small number of support vectors.


\paragraph*{Related work.}
In comparison with the recent provable algorithms for autoencoding of
Arora \etal \cite{agm-nalio-14}, our result does not require any
distributional assumptions on the $\comb$'s or $\obs$'s, e.g., that
the $\obs \in \PntSet$ were produced by choosing $\comb$ from a
particular distribution and then computing $\Mat\comb$ and adding
random noise.  It also does not require that the columns of $\Mat$ be
incoherent (nearly orthogonal).  However, we do require that the
columns of $\Mat$ be convex combinations of the points
$\obs \in \PntSet$ and that they can approximately reconstruct the
$\obs \in \PntSet$ via convex combinations, so our results are
incomparable to those of Arora \etal \cite{agm-nalio-14}.  Work on
related encoding or dictionary learning problems in the full rank case
has been done by Spielman \etal \cite{sww-ersud-12}, and efficient
algorithms for finding minimal and sparse Boolean representations
under anchor-set assumptions were given by Balcan \etal
\cite{bbv-erlla15}.

\subsection{The results in detail}

Our results are summarized in \figref{our:results}.

\smallskip%
\begin{compactenum}[(A)]
    \item \textbf{Sparse nearest-neighbor in high dimensions.} %
    For a set of points $\PntSet$ in the unit ball
    $\ballD \subseteq \Re^d$ and any point of
    $\pnt \in \CHX{\PntSet}$, one can find a point
    $\pnt' \in \CHX{\PntSet}$ that is the convex combination of
    $O(1/\eps^2)$ points of $\PntSet$, such that
    $\distY{\pnt}{\pnt'} \leq \eps$.  This is of course well known by
    now \cite{c-csgaf-10}, and we describe (for the sake of
    completeness) the surprisingly simple iterative algorithm (which
    is similar to the Perceptron algorithm) to compute such a
    representation in \secref{sparse:approx:ch}.
    This sparse representation is sometimes referred to as an
    approximate \Caratheodory theorem \cite{b-anedb-15}, and it also
    follows from the analysis of the Perceptron algorithm
    \cite{n-ocpp-62} -- see \remref{preceptron}.

    \smallskip
    \item \textbf{Geometric hitting set.}  %
    Our problem can be interpreted as (a somewhat convoluted)
    geometric hitting set problem. In particular, one can apply the
    Clarkson \cite{c-apca-93} polytope approximation algorithm to this
    problem, thus yielding an $O( d \log \kopt)$ approximation. For
    the sake of completeness, we describe this in detail in
    \secref{v:c:approx}. (Since $d$ might be large, this approximation
    is somewhat less attractive.)

    \smallskip
    \item \textbf{The greedy approach.} %
    A natural approach is to try and solve the problem using the
    greedy algorithm. Here, this requires some work, and the resulting
    algorithm is a combination of the algorithm from (A) with greedy
    set cover for the ranges defined in (B).  We initialize an
    instance of the algorithm from (A) for each point
    $\obs \in \PntSet$ whose job is to either find a hyperplane
    through $\obs$ separating it from $\PntSet \setminus\{\obs\}$ by a
    large margin or else to approximate $\obs$ as a combination of a
    few support-vectors in $\PntSet \setminus \{\obs\}$.  At each
    step, we find the point $\obs' \in \PntSet$ that causes as many of
    these algorithms to perform an update as possible, and add it into
    our set $\algset$.  The key issue is to prove that the procedure
    halts after a limited number of steps.  This algorithm is
    described in \secref{greedy:h:set}.

    \smallskip
    \item \textbf{Using greedy clustering.} %
    The second algorithm, and our main contribution, is more similar
    in spirit to the Gonzalez algorithm for $k$-center clustering:
    Repeatedly find the point $\obs \in \PntSet$ that is farthest from
    the convex hull of the points of $\algset$ and then add it into
    $\algset$ if this distance is greater than some threshold (a
    similar idea was used for subspace approximation \cite[Lemma
    5.2]{hv-hdsfl-04}).  The key issue here is to prove that some
    measure of significant progress is made each time a new point is
    added.  Somewhat surprisingly, after $O(\kopt/\eps^{2/3})$
    iterations, the resulting set is an
    $O\pth{\eps^{1/3}}$-approximation to the original set of
    points. Note, that unlike the other results mentioned above, there
    is no dependence on the dimension or the input size.

\end{compactenum}
\medskip%
An additional property of all the above algorithms is that the points
$\algset$ found will be actual dataset points and the algorithms only
require dot-product access to the data.  This means that the
algorithms can be kernelized.  Additionally, much as with CUR
decompositions of matrices, since the points $\algset$ are data
points, they will preserve sparsity if the dataset $\PntSet$ was
sparse.

\section{Preliminaries}

For a set $X\subseteq \Re^d$, $\CHX{X}$ denotes the \emph{convex hull}
of $X$.  For two sets $\PntSet, \PntSet' \subseteq \Re^d $, we denote
by
\begin{math}
    \distSet{\PntSet}{\PntSet'}%
    =%
    \min_{\pnt \in \PntSet} \min_{\pnt' \in \PntSet'}
    \distY{\pnt}{\pnt'}
\end{math}
the \emph{distance} between $\PntSet$ and $\PntSet'$. For a point
$\query \in \Re^d$, its distance to the set $\PntSet$ is
\begin{math}
    \distSet{\query}{\PntSet} = \distSet{\brc{\query}}{\PntSet},
\end{math}
and its \emphi{projection} or \emphi{nearest neighbor} in $\PntSet$ is
the point
\begin{math}
    \nnY{\query}{\PntSet}%
    =%
    \arg \min_{\pnt \in \PntSet} \distY{\query}{\pnt}.
\end{math}

\subsection{Sparse convex-approximation: %
   Problem statement and background}
   
For a set $Y$ in $\Re^d$, its \emph{one sided Hausdorff distance} from
$X$ is $\dOneSideH{Y}{X} = \max_{y \in Y} \distSet{y}{X}$. 

\begin{defn}
    \deflab{one_sided_approx}%
    Consider two sets $\Pin, \Pout \subseteq \Re^d$.
    A set $\PntSetA \subseteq \CHX{\Pout}$ is a
    \emphi{$\delta$-approximation} to $\Pin$ from $\Pout$ if
    \begin{math}
        \dOneSideH{\CHX{\Pin}}{\bigl.\CHX{\PntSetA}} \leq \delta.
    \end{math}
    In words, every point of $\CHX{\Pin}$ is within distance 
    $\delta$ from a point of $\CHX{\PntSetA}$. In the \emphi{discrete
       $\delta$-approximation} version, we require that
    $\PntSetA \subseteq {\Pout}$.
    We use $\OptAprxZ{\Pin}{\Pout}{\delta}$ to denote any minimum
    cardinality discrete $\delta$-approximation to $\Pin$ from $\Pout$,
    and
    $\koptZ{\Pin}{\Pout}{\delta} = \cardin{
       \OptAprxZ{\Pin}{\Pout}{\delta} }$
    to denote its size.  We drop the phrase ``from $\Pout$'' when it
    is clear from the context.
\end{defn}

\begin{problem}
    \problab{in_and_out}%
    Given sets $\Pin,\Pout \subseteq \Re^d$, compute (or approximate)
    $\OptAprxZ{\Pin}{\Pout}{\delta}$.
\end{problem}

For the majority of the paper we focus on the natural special case
when $\PntSet = \Pin = \Pout$.  The \emphi{Hausdorff distance} between
sets $X$ and $Y$ is defined as
\begin{math}
    \distH{X}{Y}%
    =%
    \max\bigl( \dOneSideH{Y}{X}, \dOneSideH{X}{Y} \bigr).
\end{math}

\begin{lemma}
    \lemlab{useful}%
    \begin{inparaenum}[(i)]
        \item Let $\CSet$ be a convex-set in $\Re^d$, then the
        function $f(\pnt) = \distSet{\pnt}{\CSet}$ is convex, where
        $\pnt \in \Re^d$.

        \item A convex-function $f$, over a convex bounded domain
        $\DSet \subseteq \Re^d$, attains its maximum in a boundary
        point of $\DSet$.

        \item For bounded point sets
        $\PntSetA, \PntSet \subseteq \Re^d$, such that
        $\PntSetA \subseteq \CHX{\PntSet}$, we have
        $\distH{\CHX{\PntSetA}}{\CHX{\PntSet}}
        =\dOneSideH{\PntSet}{\CHX{\PntSetA}}$.
    \end{inparaenum}
\end{lemma}
\begin{proof}
    This is all well known, and we include the proof for the sake of
    completeness.
    
    (i) Consider any two points $\pnt, \pntB$ in $\Re^d$, and let
    $\pnt' = \nnY{\pnt}{\CSet}$ and $\pntB' = \nnY{\pntB}{\CSet}$. For
    any $t \in [0,1]$, we have by convexity that
    \begin{math}
        \pntC = t \pnt + (1-t)\pntB  \in \pnt \pntB
    \end{math}
    and 
    \begin{math}
        \pntC' = t \pnt' + (1-t)\pntB' \in \CSet.
    \end{math}
    Therefore, by the triangle inequality, we have
    \begin{align*}
      f(\pntC)%
      &= %
        f\pth{ t \pnt + (1-t)\pntB \bigl.}%
        \leq %
        \distY{\pntC}{\pntC'}%
        =%
        \distY{ \bigl. \pth{ t \pnt + (1-t)\pntB \bigl.}%
        }{ \pth{ t \pnt' + (1-t)\pntB' \bigl.}}%
      \\&%
      =%
      \norm{ \bigl. t (\pnt-\pnt') + (1-t)(\pntB - \pntB')}%
      \leq%
      \norm{ \bigl. t (\pnt-\pnt')} + \norm{\bigl.(1-t)(\pntB -
      \pntB')}%
      \\&%
      =%
      t \norm{ \bigl. \pnt-\pnt'} + (1-t)\norm{\bigl.\pntB -
      \pntB'}%
      =%
      t f(\pnt) + (1-t)f(\pntB).
    \end{align*}

    (ii) If $\pnt$ is the interior of $\DSet$ then there are extremal
    points $\pnt_1, \ldots, \pnt_d$ of $\DSet$, and constants
    $\alpha_1, \ldots, \alpha_d \in [0,1]$, such that
    $\sum_i \alpha_i = 1$ and $\pnt =\sum_i \alpha_i \pnt_i$. As such,
    by convexity, we have
    \begin{math}
        f(\pnt)%
        =%
        f(\sum_i \alpha_i \pnt_i)%
        \leq%
        \sum_i \alpha_i f( \pnt_i) \leq%
        \max_i f(\pnt_i).
    \end{math}

    (iii) By (i), the function $\distSet{\pnt}{\CHX{\PntSetA}}$ is
    convex. By (ii), its maximum over $\CHX{\PntSet}$ is attained at a
    point of $\PntSet$. We thus have that
    \begin{align*}%
      \distH{\bigl.\CHX{\PntSetA}}{\CHX{\PntSet}}%
      &=%
        \max\pth{ \bigl.%
        \dOneSideH{\CHX{\PntSetA}}{\CHX{\PntSet}}, %
        \dOneSideH{\CHX{\PntSet}}{\CHX{\PntSetA}} } %
        =%
        \dOneSideH{\bigl.\CHX{\PntSet}}{\CHX{\PntSetA}} %
        =%
        \max_{\pnt \in \CHX{\PntSet}} \distSet{\pnt}{\CHX{\PntSetA}}
        =%
        \max_{\pnt \in \PntSet} \distSet{\pnt}{\CHX{\PntSetA}}%
      \\&%
      =%
      \dOneSideH{\PntSet}{\CHX{\PntSetA}}.
    \end{align*}
\end{proof}%

\begin{defn}
    Consider any set $\PntSet \subseteq \Re^d$.  A set
    $\PntSetA \subseteq \CHX{\PntSet}$ is a
    \emphi{$\delta$-approximation} to $\PntSet$ if
    \begin{math}
        \distH{\bigl.\CHX{\PntSetA}}{\CHX{\PntSet}} \leq
        \delta. 
    \end{math}
    By the above lemma, this is equivalent to every point of $\PntSet$
    being in distance at most $\delta$ from a point of
    $\CHX{\PntSetA}$.  
    In the \emphi{discrete $\delta$-approximation}
    version, we require that $\PntSetA \subseteq {\PntSet}$.  
    Let $\OptAprxY{\PntSet}{\delta}$ be any minimum cardinality
    $\delta$-approximation to $\PntSet$, and let
    $ \koptY{\PntSet}{\delta} = \cardin{ \OptAprxY{\PntSet}{\delta} }$
    denote its size.
\end{defn}

\begin{problem}
    \problab{in_equals_out} Given a set $\PntSet \subseteq \Re^d$ and
    value $\delta$, compute (or approximate)
    $\OptAprxY{\PntSet}{\delta}$.
\end{problem}

\begin{example:unnumbered}
    Consider a unit radius sphere $\Sphere{d-1}$ in $\Re^d$ centered
    at the origin, and let $\PntSet$ be a $\delta'$-packing on
    $\Sphere{d-1}$ (i.e., every point in $\Sphere{d-1}$ is at distance
    at most $\delta'$ from a point of $\PntSet$, and any two points of
    $\PntSet$ are at distance at least $\delta'$ from each other). It
    is easy to verify that such a $\delta'$-packing has size
    $\Theta\pth{ 1/(\delta')^{d-1}}$.  Furthermore, for any
    $\delta > 0$, and an appropriate absolute constant $c$
    (independent of the dimension or $\delta$), setting
    $\delta' = c \sqrt{\delta}$, we have the property that for any
    point $\pnt \in \PntSet$,
    \begin{math}
        \distSet{\pnt}{\CHX{\PntSet \setminus \brc{\pnt}}} > \delta.
    \end{math}
    That is, any $\delta$-approximation to $\PntSet$ requires
    $\Omega\pth{ 1/\delta^{(d-1)/2}}$ points.

    On the other hand, let
    $\Pout = \Set{ \pm d \eVec{i}}{i=1,\ldots,d \bigr.}$, where
    $\eVec{i}$ denotes the $i$\th orthonormal vector, having zero in
    all coordinates except for the $i$\th coordinate where it is
    $1$. Clearly, $\Sphere{d-1} \subseteq \CHX{\Pout}$, and as such
    $\koptZ{\PntSet}{\Pout}{\delta} \leq \cardin{\Pout} = 2d$, with
    equality for $\delta=0$.

    Throughout this paper we require that $\Pout$ be contained in the
    unit ball, disallowing this latter type of ``trivial'' solution,
    and furthermore having the property that a successful
    approximation also yields a {\em sparse} solution essentially for
    free, as shown next in \lemref{n:n}.
\end{example:unnumbered}

\subsection{Computing the approximate distance %
   to the convex hull}
\seclab{sparse:approx:ch}%

The following is well known, and is included for the sake of
completeness, see \cite{hkmr-seps-15}.  It also follows readily from
the Preceptron algorithm (see \remref{preceptron} below).

\begin{lemma}
    \lemlab{n:n}%
    Let $\PntSet \subseteq \Re^d$ be a point set, $\eps > 0$ be a
    parameter, and let $\query \in \Re^d$ be a given query
    point. Then, one can compute, in
    $O\pth{\cardin{\PntSet} d/\eps^2}$ time, a point
    $\pntA\in \CHX{\PntSet}$, such that
    \begin{math}
        \distY{\query}{\pntA} \leq \distSet{\query}{\CHX{\PntSet}} +
        \eps \diam,
    \end{math}
    where $\diam = \diamX{\PntSet}$. Furthermore, $\pntA$ is a convex
    combination of $O(1/\eps^2)$ points of $\PntSet$.
\end{lemma}

\begin{proof}
    The algorithm is iterative, computing a sequence of points
    $\pntA_0, \ldots, \pntA_i$ inside $\CHX{\PntSet}$ that approach
    $\query$. Initially, $\pnt_0 = \pntA_0$ is the closest point of
    $\PntSet$ to $\query$. In the $i$\th iteration, the algorithm
    computes the vector
    \begin{math}
        \vect_i = \query - \pntA_{i-1},
    \end{math}
    and the point $\pnt_i \in \PntSet$ that is extremal in the
    direction of $\vect_i$. Now, the algorithm sets $\pntA_i$ to be
    the closest point to $\query$ on the segment
    $\seg_i = \pntA_{i-1} \pnt_i$, and continues to the next
    iteration, for $\niter = O(1/\eps^2)$ iterations. The algorithm
    returns the point $\pntA^{}_{\niter}$ as the desired answer.

    \InNotJVer{%
    \parpic[r]{%
       \begin{minipage}{3.7cm}%
           \includegraphics{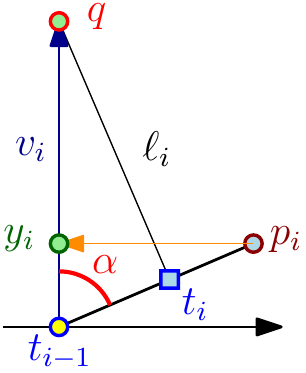}%
           \vspace{-0.8cm}%
           \captionof{figure}{}%
           \figlab{proof:n:n}%
       \end{minipage}%
    }%
    }
    \InJVer{%
       \begin{figure}[h]
           {\includegraphics{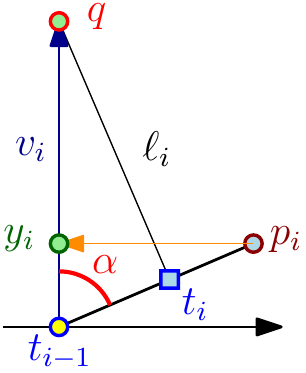}}%
           \captionof{figure}{}%
           \figlab{proof:n:n}%
       \end{figure}%
    }

    By induction, the point
    $\pntA_i \in \CHX{ \brc{\pnt^{}_0, \ldots, \pnt^{}_i}}$.
    Furthermore, observe that the distance of the points
    $\pntA_0, \pntA_1, \ldots$ from $\query$ is monotonically
    decreasing. In particular, for all $i > 0$, $\pntA_i$ must fall in
    the middle of the segment $\seg_i$, as otherwise, $\pnt_i$ would
    be closer to $\query$ than $\pnt^{}_0$, a contradiction to the
    definition of $\pnt^{}_0$.
    
    Project the point $\pnt_i$ to the segment $\pntA_{i-1} \query$,
    and let $\pntB_i$ be the projected point. Observe that
    $\distY{\query}{\pntB_i}$ is a lower bound on
    $\distSet{\query}{\CHX{\PntSet}}$. Therefore, if
    \begin{math}
        \distY{\pntB_i}{\pntA_{i-1}} \leq \eps \diam
    \end{math}
    then we are done, as
    \begin{math}
        \distY{\query}{\pntA_{i-1}}%
        \leq%
        \distY{\pntA_{i-1}}{\pntB_i} + \distY{\pntB_{i}}{\query}%
        \leq%
        \eps \diam + \distSet{\query}{\CHX{\PntSet}}.
    \end{math}
    (In particular, one can use this as alternative stopping condition
    for the algorithm, instead of counting iterations.)

    So, let $\alpha$ be the angle $\angle \pnt_i \pntA_{i-1} \query$.
    Observe that as $\pntA_{i-1} \pnt_i \subseteq \CHX{\PntSet}$, it
    follows that
    \begin{math}
        \distY{\pntA_{i-1}}{\pnt_i} \leq \diamX{\PntSet} = \diam.
    \end{math}
    Furthermore,
    \begin{math}%
        \ds%
        \cos \alpha%
        =%
        \frac{\distY{\pntB_i}{\pntA_{i-1}}}{\distY{\pntA_{i-1}}{\pnt_i}}%
        >%
        \frac{\eps \diam}{\diam}%
        =%
        \eps,
    \end{math}
    since $\distY{\pntB_i}{\pntA_{i-1}} > \eps \diam$.  Hence,
    \begin{math}
        \sin \alpha = \sqrt{1-\cos^2 \alpha}%
        \leq %
        \sqrt{1-\eps^2} \leq 1-\eps^2/2.
    \end{math}
    Let $\ell_{i-1} = \distY{\query}{\pntA_{i-1}}$. We have that
    \begin{align*}
        \ell_i%
        =%
        \distY{\query}{\pntA_i}%
        =%
        \distY{\query}{\pntA_{i-1}} \sin \alpha \leq (1-\eps^2/2)
        \ell_{i-1}.
    \end{align*}

    Analyzing the number of iterations required by the algorithm is
    somewhat tedious. If
    \begin{math}
        \ell_0 = \distY{\query}{\pntA_0} \geq (4/\eps^2) \diam
    \end{math}
    then the algorithm would be done in one iteration as otherwise
    $\ell_1 \leq \ell_0 -2\diam$, which is impossible. In particular,
    after $4/\eps^2$ iterations the distance $\ell_i$ shrinks by a
    factor of two, and as such, after $O((1/\eps^2) \log (1/\eps))$
    iterations the algorithm is done.

    One can do somewhat better. By the above, we can assume that
    $\distSet{\query}{\PntSet} = O(\diam / \eps^2)$.  Now, set
    $\eps_j = 1/2^{2+j}$. By the above, after
    $n_0 = O((1/\eps_0^2) \log (1/\eps_0)) =O(1)$ iterations,
    $\ell_{n_0} \leq \distSet{\query}{\CHX{\PntSet}} +
    \diamX{\PntSet}/4$.
    For $j \geq 1$, let $n_j = 4/(\eps_{j})^2$, and observe that,
    after $\nu_j = n_j + \sum_{k=0}^{j-1} n_k$ iterations, we have
    that
    \begin{align*}
        \ell_{\nu_j}%
        \leq%
        \pth{\bigl. \distSet{\query}{\CHX{\PntSet}} + \eps_{j-1}
           \diam}/2%
        \leq%
        \distSet{\query}{\CHX{\PntSet}} + \eps_{j} \diam.
    \end{align*}
    In particular, stopping as soon as $\eps_j \leq \eps$, we have the
    desired guarantee, and the number of iterations needed is
    $\niter = O(1) + \sum_{j=0}^{\ceil{\lg 1/\eps}} 4/\eps_j^2 =
    O(1/\eps^2)$.
\end{proof}%

In our use of \lemref{n:n}, $\PntSet$ and $\query$ will always be
contained in the unit ball, so we can remove the $\diam$ term in the
bound if we wish since $\diam \leq 2$.

\begin{remark}
    \remlab{preceptron}%
    \lemref{n:n} is known, and a variant of it follows readily from a
    result (from 1962) on the convergence of the Perceptron algorithm
    \cite{n-ocpp-62}.  Indeed, consider a set
    $\PntSet \subseteq \Re^d$ and a query point $\query \in \Re^d$.
    Assume that $\query \in \CHX{\PntSet}$, and furthermore that
    $\query$ is the origin (translating space if needed to ensure
    this).  Run the Perceptron algorithm learning a linear classifier
    that passes through the origin and classifies $\PntSet$ as
    positive examples.  Stop the algorithm after $M = 1/\eps^2$
    classification mistakes (since $\query \in \CHX{\PntSet}$, there
    will always be a mistake in $\PntSet$).  Let
    $\pnt^{}_1, \ldots, \pnt^{}_M$ be the sequence of points on which
    mistakes were made and let $w = \pnt_1 + \ldots + \pnt_M$ be the
    resulting hypothesis vector.  By the analysis of \cite{n-ocpp-62},
    we have $\norm{w} \leq \diamX{\PntSet}\sqrt{M}$.  This implies
    that the point $\pnt' = w/M$, which is a convex combination of the
    points $\pnt_1, \ldots, \pnt_M$, has length---and therefore
    distance from $\query$---at most $\eps \diamX{\PntSet}$.

    Thus, we conclude that for any point $\pnt \in \CHX{\PntSet}$, and
    any $\eps \in (0,1)$, there is a point $\pnt' \in \CHX{\PntSetA}$,
    which is a convex combination of $O(1/\eps^2)$ points of
    $\PntSet$, such that
    $\distY{\pnt}{\pnt'} \leq \eps \diamX{\PntSet}$. This is sometimes
    referred to as approximate \Caratheodory theorem
    \cite{b-anedb-15}.

    We described the alternative algorithm (in the proof of
    \lemref{n:n}) because it is more direct and slightly simpler in
    this case.
\end{remark}



\section{Approximations via hitting set algorithms}
\seclab{hitting:sets}

Here we look at two hitting set type algorithms for
\probref{in_and_out}. An
\emphi{$\ApproxQB{\alpha}{\beta}$-approximation} of
$\OptAprxZ{\Pin}{\Pout}{\eps}$ is a set $\PntSetA \subseteq {\Pout}$
such that
\begin{math}
    \dOneSideH{\CHX{\Pin}}{\bigl.\CHX{\PntSetA}} \leq \alpha
\end{math}
and
\begin{math}
    \cardin{\PntSetA} \leq \beta \koptZ{\Pin}{\Pout}{\eps},
\end{math}%
see \defref{one_sided_approx}.

As a warm-up exercise, we first present an
$(\eps,\allowbreak O( d \log \kopt ))$-approximation using
approximation algorithms for hitting sets for set systems with bounded
\VC dimension. Then, we build on that to get a greedy algorithm
providing a
$( (1+\delta)\eps, O((\eps\delta)^{-2} \log n) )$-approximation.

\subsection{Approximation via \VC dimension}
\seclab{v:c:approx}%

\noindent%
\begin{minipage}{0.84\linewidth}
\begin{defn}
    \deflab{shadow}%
    For a set $\PntSet \subseteq \Re^d$ and a direction vector
    $\vect$, let $\pnt$ be the point of $\PntSet$ extreme in the
    direction of $\vect$, and let $\hplane'$ be the hyperplane with
    normal $\vect$ and tangent to $\CHX{\PntSet}$ at $\pnt$.  For a
    parameter $\eps$, let $\hplane$ be the hyperplane formed by
    translating $\hplane'$ distance $\eps$ in the direction $-\vect$.
    The \emphi{$\eps$-shadow} of $\hplane'$ (or $\vect$), is the
    halfspace $\hpZ{\PntSet}{\eps}{\vect}$ bounded by $\hplane$ that
    contains $\pnt$ in its interior. In words, the $\eps$-shadow of
    $\vect$ is the outer supporting halfspace for $\PntSet$ with a
    normal in the direction of $\vect$, translated in by distance
    $\eps$.
\end{defn}
\end{minipage}\hfill
\begin{minipage}{0.16\linewidth}%
    \hfill%
    {\includegraphics{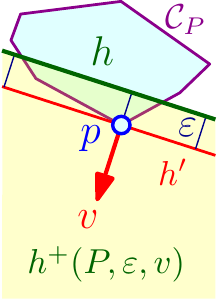}}%
\end{minipage}

\begin{lemma}
    \lemlab{v:c:argument}%
    Given sets $\Pin$ and $\Pout$ in $\Re^d$ with a total of $n$
    points, and a parameter $\eps$, one can compute a
    $\ApproxQB{\eps}{ O( d \log \kopt )\bigr.}$-approximation to the
    optimal discrete set $\OptAprxZ{\Pin}{\Pout}{\eps}$ in polynomial
    time.
\end{lemma}

\begin{proof}
    For a direction $\vect$, consider the hyperplane $\hplane'$
    tangent to $\CHX{\Pin}$ at an extremal point
    $\pnt_{\vect} \in \Pin$ in the direction of $\vect$, and its
    $\eps$-shadow $\hplane^+ =\hpZ{\Pin}{\eps}{\vect}$, see
    \figref{s:fig}.%

    \InNotJVer{%
    \parpic[r]{%
       \begin{minipage}{5cm}%
           ~\hfill%
           \begin{minipage}{4.7cm}
               \hfill{\includegraphics[width=0.95\linewidth]{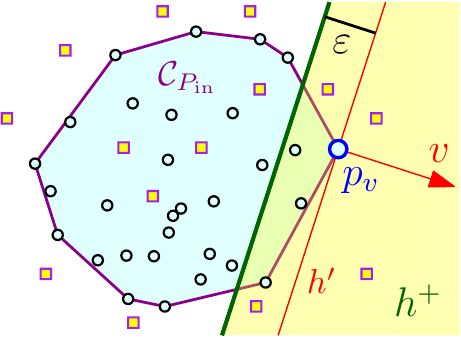}}%
               \captionof{figure}{Circles and squares denote points of
                  $\Pin$ and $\Pout$, respectively.}%
               \figlab{s:fig}%
           \end{minipage}%
       \end{minipage}%
    }%
    }
    \InJVer{%
       \begin{figure}[h]
           \includegraphics{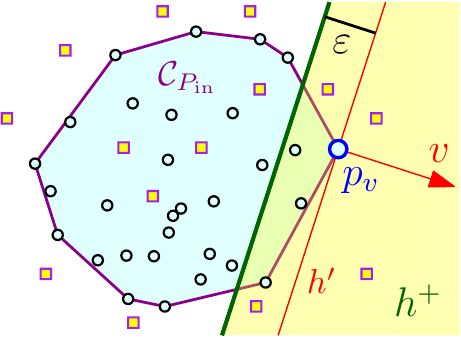}               
           \captionof{figure}{Circles and squares denote points of
              $\Pin$ and $\Pout$, respectively.}%
           \figlab{s:fig}%
       \end{figure}%
    }%
    
    Clearly, any discrete $\eps$-approximation
    $\PntSetA \subseteq \Pout$ to $\Pin$, must contain at least one
    point of $\Pout \cap \hplane^+$, as otherwise the approximation
    fails for the point $\pnt_{\vect}$ (in particular, if such a
    halfspace has no point in $\Pout$ then there is no
    approximation). Now, consider the set system
    \begin{align*}
        \RangeSpace%
        =%
        \pth{\Bigl.\Pout,%
           \Set{%
              \Pout \cap \hpZ{\Pin}{\eps }{\vect}%
           }{\vect \text{ any unit vector}}}.
    \end{align*}
    This set system has \VC dimension at most $d+1$, and in
    particular, for such a set system one can compute a
    $O( d \log \kopt )$ approximation to its minimum size hitting set,
    which is the desired approximation in this case, see \cite[Section
    6.3]{h-gaa-11}.  We describe the algorithm below, but first we
    verify that this indeed yields the desired approximation.

    \InNotJVer{%
       \parpic[r]{%
          \begin{minipage}{5cm}%
              ~\hfill%
              \begin{minipage}{4.7cm}
                  \hfill{\includegraphics{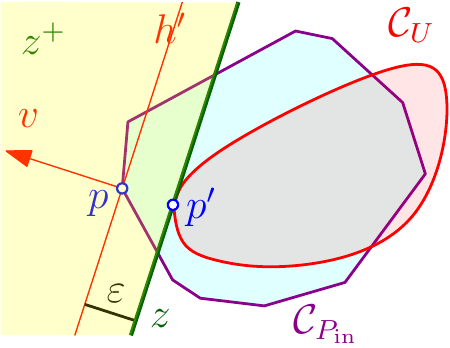}}%
                  \captionof{figure}{}%
                  \figlab{goody}%
              \end{minipage}%
          \end{minipage}%
       }%
    }%
    \InJVer{%
       \begin{figure}[h]
           \includegraphics{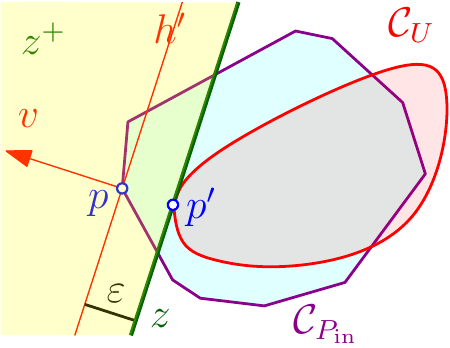}%
           \captionof{figure}{}%
           \figlab{goody}%
       \end{figure}
    }%
 
    Consider a hitting set $\PntSetA \subseteq \Pout$ of
    $\RangeSpace$. Let $\pnt$ be any point in $\CHX{ \Pin}$, and let
    $\pnt'$ be the closest point to $\pnt$ in $\CHX{\PntSetA}$. If
    $\distY{\pnt}{\pnt'} \leq \eps $, then we are done. Otherwise,
    consider the vector $\vect = \pnt - \pnt'$.  Let $\hplaneA$ denote
    the hyperplane whose normal is $\vect$ and which passes through
    the point $\pnt'$, and let $\hplaneA^+$ denote the open halfspace
    bounded by $\hplaneA$ and in the direction of $\vect$
    (i.e. containing $\pnt$).  As $\pnt'$ is the closest point to
    $\pnt$ in $\CHX{\PntSetA}$, $\hplaneA^+$ has empty intersection
    with $\CHX{\PntSetA}$.  Moreover,
    $\hpZ{\Pin}{\eps }{\vect}\subsetneq \hplaneA^+$, as the bounding
    hyperplanes of both halfspaces have $\vect$ as a normal, and the
    extreme point of $\CHX{ \Pin}$ in the direction of $\vect$ must be
    $>\eps$ away from $\hplaneA$ (as $\pnt$ is at least this far in
    the direction of $\vect$).  See \figref{goody}.  These two facts
    combined imply
    $\hpZ{\Pin}{\eps }{\vect}\cap \CHX{\PntSetA} = \emptyset$, a
    contradiction as $\hpZ{\Pin}{\eps }{\vect} \cap \Pout$ is a set in
    $\RangeSpace$ that should have been hit.

    As for the algorithm, Clarkson \cite{c-apca-93} described how to
    compute this set via reweighting, but the following technique due
    to Long \cite{l-upsda-01} is easier to describe (we sketch it here
    for the sake of completeness). Consider the \LP relaxation of the
    hitting set for this set system. Clearly, one can assign weights
    to points (between $0$ and $1$), such that the total weight of the
    points is at most $\kopt$, and for every range in $\RangeSpace$
    the total weight of the points it covers is at least $1$. Dividing
    this fractional solution by $\kopt$, we get a weighted set system,
    where every set has weight at least $\eta = 1/\kopt$, and total
    weight of the points is $1$. That is, we can interpret these
    weights over the points as a measure, where all the sets of
    interests are $\eta$-heavy. A random sample of size
    $O( (d/\eta) \log (1/\eta)) = O\pth{\kopt d \log \kopt}$ of
    $\PntSet$ (according to the weights) is an $\eta$-net with
    constant probability \cite{hw-ensrq-87}, and stabs all the sets of
    $\RangeSpace$, as desired. Should the random sample fail, one can
    sample again till success.
\end{proof}

\subsection{Approximation via a greedy algorithm}
\seclab{greedy:h:set}%

\begin{lemma}
    \lemlab{greedy}%
    Let $\Pin$ and $\Pout$ be sets of points in $\Re^d$ contained in
    the unit ball, with a total of $n$ points.  For parameters
    $\eps, \delta \in (0,1)$, one can compute, in polynomial time, a
    $\ApproxQB{\bigl.(1+\delta)\eps}%
    { O( \eps^{-2} \delta^{-2} \log n )}$-approximation
    to the optimal discrete set
    $\OptAprxZ{\Pin}{\Pout}{\eps }$.
\end{lemma}

\begin{proof}
    The algorithm is greedy -- the basic idea is to restrict the set
    system of \lemref{v:c:argument} to the relevant active
    sets. Formally, let $\PntSetA_0 = \brc{\pnt_0}$, where $\pnt_0$ is
    some arbitrary point of $\Pout$. For $i>0$, in the $i$\th
    iteration, consider the current convex set
    $\CSet_{i-1} = \CHX{ \PntSetA_{i-1}}$.  For a point
    $\query \in \Pin \setminus \CSet_{i-1}$, let
    $\nnY{\query}{C_{i-1}}$ be its nearest point in $C_{i-1}$, and let
    $\vect_i(\query)$ be the direction of the vector
    $\query - \nnY{\pnt}{C_{i-1}}$.  In particular, consider the
    $\eps$-shadow halfspace
    $\hplane^+ = \hpZ{\Pin}{\eps }{\vect_i(\query)}$, see
    \defref{shadow}, which should be hit by the desired hitting
    set\footnote{The hitting set computed by the algorithm is somewhat
       weaker, only hitting all the $(1+\delta)\eps$-shadows.}.
    
    Let $\PntSetC_i \subseteq \Pin$ be the set of points of $\Pin$
    that are \emph{unhappy}; that is, they are in distance
    $\geq (1+\delta)\eps$ from $\CHX{\PntSetA_{i-1}}$.  We restrict
    our attention to the set system of active halfspaces; that is,
    \begin{align*}
      \RangeSpace_i%
      =%
      \pth{\Bigl.\Pout,%
      \Set{%
      \Pout \cap \hpZ{\bigl.\Pin}{\eps
      }{\vect_i(\query)}%
      }{%
      \query \in \PntSetC_i
      }%
      }.%
    \end{align*}
    (As before, if $\Pout\cap \hplane^+ $ is empty, then no
    approximation is possible, and the algorithm is done.)  Now, as in
    the classical algorithm for hitting set (or set cover), pick the
    point $\pnt_i$ in $\Pout$ that hits the largest number of ranges
    in $\RangeSpace_i$, and add it to $\PntSetA_{i-1}$ to form
    $\PntSetA_i$.

    A point $\query \in \PntSetC_i$, is \emph{hit} in the $i$\th
    iteration if
    \begin{math}
        \pnt_{i} \in \hpZ{\bigl.\PntSet}{\eps
           }{\vect_i(\query)}.
    \end{math}
    The argument of \lemref{n:n} (or \remref{preceptron}) implies that
    after a point $\query \in \Pin$ is hit $c/(\eps^2\delta^2)$ times,
    its distance to the convex-hull of the current points is smaller
    than $(1+\delta)\eps$, and it is no longer unhappy, where $c$ is
    some sufficiently large constant.  Indeed, using the notation of
    the proof \lemref{n:n}, if a point $\query \in \PntSetC_i$ is hit
    in the $i$\th iteration by a point $\pnt_i$, and
    $\distSet{\query}{\CHX{\PntSetA_{i-1}}} \leq (1+\delta)\eps$ then
    we are done. Otherwise, let
    $\pntA_{i-1}= \nnY{\query}{\CHX{\PntSetA_{i-1}}}$, and let
    $\pntB_i$ be the projection of $\pnt_i$ to the segment
    $\query \pntA_{i-1}$, see \figref{proof:n:n}. We have that
    \begin{math}
        \distY{\pntB_i}{\pntA_{i-1}}%
        \geq%
        \distY{\query}{\pntA_{i-1}} - \distY{\query}{\pntB_{i}}%
        \geq%
        (1+\delta)\eps - \eps%
        \geq%
        \eps \delta,
    \end{math}
    since $\distY{\query}{\pntB_{i}} \leq \eps$ (as $\pnt_i$ and
    $\pntB_i$ are both in the $\eps$-shadow of $\query$). Now, the
    analysis of \lemref{n:n} applies (with $\eps \delta$ instead of
    $\eps$), implying that after $O(1/(\eps \delta)^2)$ iterations,
    the distance of $\query$ from the current convex-hull would be
    smaller than $(1+\delta)\eps$.

    So, let $n_i$ be the number of unhappy points in the beginning of
    the $i$\th iteration, and observe that at least $n_i/\kopt$ points
    are being hit in the $i$\th iteration. In particular, let
    $\kappa = 2\ceil{c \kopt / (\eps^2\delta^2)}$, and observe that in
    the iterations between $i-\kappa$ and $i$, we have that the number
    of points being hit is at least
    \begin{math}
        \sum_{j=i-\kappa}^i n_j / \kopt \geq 2 n_i c /(\eps^2\delta^2).
    \end{math}
    This implies that $n_{i-\kappa} \geq 2n_i$. Otherwise,
    $n_{i-\kappa} < 2n_i$, implying that in this range of iterations
    $> N = n_{i-\kappa} c /(\eps^2\delta^2)$ hits happened, which is
    impossible, as $n_{i-\kappa}$ points can be hit at most $N$ times
    before they are all happy.

    As such, after $\kappa$ iterations of the greedy algorithm, the
    number of unhappy points drops by a factor of two, and we conclude
    that after $O( \kopt \, (\eps\delta)^{-2} \log n)$ total
    iterations, the algorithm is done.
\end{proof}


\section{Approximating the convex hull in high %
   dimensions}

Here we provide an efficient bi-criteria approximation algorithm for
\probref{in_equals_out}. That is, the algorithm computes a subset
$\PntSetA \subseteq \CHX{\PntSet}$, such that 
\begin{inparaenum}[(i)]
    \item
    $\distH{\CHX{\PntSetA}}{\CHX{\PntSet}} \leq O\pth{\eps^{1/3}}
    \diamX{\PntSet}$, and
    \item
    $\cardin{\PntSetA} \leq O\pth{\koptY{\PntSet}{\eps}/\eps^{2/3}}$.
\end{inparaenum}
Significantly, the computed set $\PntSetA$ is actually a subset of $\PntSet$, 
implying that the algorithm simultaneously solves both the continuous and discrete
variants of the problem.

To simplify the presentation, in the remainder of this section we
assume $\diam = \diamX{\PntSet} = O(1)$, and hence drop most
appearances of $\diam$.

\subsection{The algorithm}

Let $\reach = 8\eps^{1/3}$. The algorithm is greedy, similar in spirit
to the Gonzalez algorithm for $k$-center clustering \cite{g-cmmid-85}
and subspace approximation algorithms \cite[Lemma 5.2]{hv-hdsfl-04}.
The algorithm starts with an arbitrary point $\pntA_0 \in \PntSet$.
For $i > 0$, in the $i$\th iteration, the algorithm computes the point
$\pntA_i$ in $\PntSet$ which is furthest away from
\begin{math}
    \CHX{\PntSetA_{i-1}},
\end{math}
where
\begin{math}
    \PntSetA_{i-1} = \brc{\pntA_0, \dots, \pntA_{i-1}}.
\end{math}
For now assume these distance queries are done exactly -- later on we
describe how to use approximate queries (i.e., \lemref{n:n}).  Let
\begin{math}
    r_i = \distSet{\pntA_i}{\CHX{\PntSetA_{i-1}}}.
\end{math}
The algorithm stops as soon as $r_i\leq \reach $, and outputs
$\PntSetA_{i-1}$.

\begin{observation}
    In the above algorithm, for all $i>0$, the point $\pntA_i$ is a
    vertex of $\CHX{\PntSet}$ (so long as exact distance queries are
    used).  In particular, if the output has to be a subset of the
    convex hull vertices, one can choose $\pntA_0$ to be the extreme
    vertex in any direction.
\end{observation}%

\subsection{Analysis}

By the termination condition of the algorithm, when the algorithm
stops every point in $\PntSet$ is in distance at most
$\reach = 8\eps^{1/3}$ away from $\CHX{\PntSetA_{i-1}}$, as desired.
As for the number of rounds until termination, we argue that in each
round there exists some point $\opnt \in \OptSet$ which is far from
$\CHX{\PntSetA_{i-1}}$ (as specified in \clmref{optBounded}) and such
that
\begin{math}
    \distSet{\opnt}{\PntSetA_{i}}%
    \leq%
    (1-\Omega(\eps^{2/3})) \distSet{\opnt}{\PntSetA_{i-1}}.
\end{math}

\InNotJVer{%
   \parpic[r]{%
      \begin{minipage}{0.45\linewidth} \hfill%
          \includegraphics{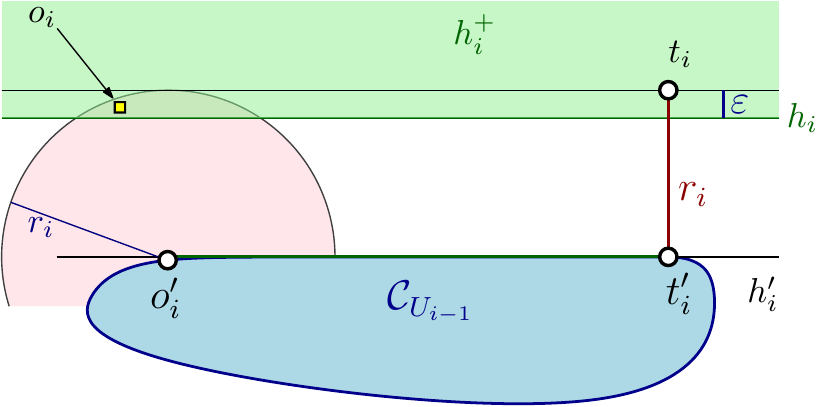}%
          \captionof{figure}{}%
          \figlab{d:2}%
      \end{minipage}%
   }%
}%
\InJVer{%
   \begin{figure}[h]
       \includegraphics{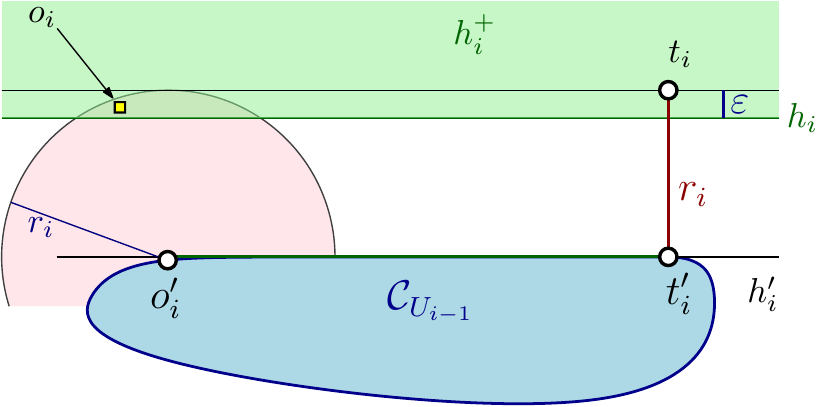}%
       \captionof{figure}{}%
       \figlab{d:2}%
   \end{figure}%
}

So consider some round $i$, the current set $\PntSetA_{i-1}$, and the
point $\pntA_i \in \PntSet$ furthest away from $\CHX{\PntSetA_{i-1}}$.
Let $\pntA_i'$ be the closest point to $\pntA_i$ in
$\CHX{\PntSetA_{i-1}}$, and let $r_i=\distY{\pntA_i}{\pntA_i'}$.  Let
$\hplane_i$ be the hyperplane orthogonal to the segment
$\pntA_i \pntA_i'$ and lying $\eps$ distance below $\pntA_i$ in the
direction of $\pntA_i'$. Let $\hplane_i^+$ denote the closed halfspace
having $\hplane_i$ as its boundary, and that contains $\pntA_i$, see
\figref{d:2}.  If no points of $\OptSet$ are in $\hplane_i^+$ then
$\distSet{\bigl.\pntA_i}{\CHX{\OptSet}} > \eps $, which is
impossible. Therefore, there must be a point
\begin{math}
    \opnt_i \in \OptSet \cap \hplane_i^+.
\end{math}
Let $\opnt_i'$ be the closest point to $\opnt_i$ in
$\CHX{\PntSetA_{i-1}}$.

\begin{claim}%
    \clmlab{optBounded}%
    $r_i-\eps \leq \distY{\opnt_i}{\opnt_i'} \leq r_i$.
\end{claim}
\begin{proof}
    Let $\hplane_i'$ be the translation of $\hplane_i$ so it passes
    through $\pntA_i'$, see \figref{d:2}.  We have that
    \begin{math}
        r_i-\eps%
        =%
        \distSet{\hplane_i'}{\hplane_i}%
        \leq%
        \distY{\opnt_i}{\opnt_i'},
    \end{math}
    as $\opnt_i$ lies in $\hplane_i^+$ (i.e., above $\hplane_i$) and
    all of $\CHX{\PntSetA_{i-1}}$ lies below $\hplane_i'$.

    For the second part, for any $\pnt \in \Re^d$, let
    $f_{i-1}(\pnt)$ be the distance of $\pnt$ from
    $\CHX{\PntSetA_{i-1}}$.  By \lemref{useful} (iii), and since
    $\opnt_i \in \OptSet \subseteq \CHX{\PntSet}$, it follows that
    \begin{math}
        \distY{\opnt_i}{\opnt_i'}%
        \leq%
        \max_{\pnt \in \CHX{\PntSet}} f_{i-1}\pth{ \pnt}%
        =%
        \distY{\pntA_i}{\pntA_{i}'}%
        =%
        r_i.
    \end{math}
\end{proof}%

\begin{lemma}
    \lemlab{shrink}%
    If $r_i \geq 8 \eps^{1/3}$ then
    \begin{math}
        \distSet{\bigl. \opnt_i}{\CHX{\PntSetA_i}}%
        \leq%
        \pth{1- \eps^{2/3}}
        \distSet{\bigl.\opnt_i}{\CHX{\PntSetA_{i-1}}}.
    \end{math}
\end{lemma}

\begin{proof}
    In the following, all entities are defined in the context of the
    $i$\th iteration, and we omit the subscript $i$ denoting this to
    simplify the exposition.  Assume, for the time being, that the
    angle $\angle \pntA \pntA' \opnt'$ is a right angle and the
    segment $\pntA' \opnt'$ has length $\ell = 1$, see
    \figref{second:stage}. This is the worst case configuration in
    terms of the new convex-hull $\CHX{\PntSetA_i}$ getting closer to
    $\opnt$, as can be easily seen.

    \InNotJVer{%
    \parpic[r]{%
       \begin{minipage}{8.5cm}%
           \centerline{\includegraphics[page=2]{figs/distances_2_ext}}%
           \captionof{figure}{}%
           \figlab{second:stage}%
       \end{minipage}
    }%
    }
    \InJVer{%
       \begin{figure}
           \centerline{\includegraphics[page=2]{figs/distances_2_ext}}%
           \captionof{figure}{}%
           \figlab{second:stage}%
       \end{figure}
    }

    Let $\pntC$ be the intersection of $\hplane$ with the ray
    emanating from $\opnt'$ in the direction $\pntA -\pntA'$. Let
    $\pntC'$ be the closest point to $\pntC$ on $\opnt' \pntA$, let
    $\idist = \distY{\pntC}{\pntC'}$, and let $\rho$ be the radius of
    the ball formed by $\ballY{\opnt'}{r} \cap \hplane$. See
    \figref{second:stage}.
    
    Rather than bounding the distance of $\opnt$ to $\CHX{\PntSetA_i}$
    directly, instead we use bounds on $\rho$ and $\idist$.  Observe
    that
    \begin{math}
        \opnt%
        \in%
        \hplane^+ \cap \ballY{\opnt'}{r}%
        \subseteq%
        \ballY{\pntC}{\rho}, 
    \end{math}
    and as such, $\distY{\opnt}{\pntC} \leq \rho$.  Now, we have
    \begin{math}
        \rho%
        =%
        \sqrt{r^2 - \distY{\pntC}{\opnt'}^2}%
        = %
        \sqrt{r^2 - (r - \eps)^2}%
        =%
        \sqrt{2 r \eps - \eps^2}%
        \leq%
        \sqrt{2 r \eps}.
    \end{math}

    Let $\alpha = \angle \pntC \opnt'\pntA$ and
    $\beta = \pi / 2 - \alpha = \angle \pntA \opnt' \pntA'$, and
    observe that
    $\sin \alpha = \cos \beta = \ell /\sqrt{\ell^2 + r^2}$, where
    $\ell = \distY{\opnt'}{\pntA'}=1$.  Now, we have
    \begin{align}
      \frac{\idist}{r - \eps}%
      &=%
        \sin \alpha%
        =%
        \frac{{\ell}} { \sqrt{\ell^2 +r^2}}%
        =%
        \frac{1}{\sqrt{1 +r^2}}%
        \leq%
        {\sqrt{1 - \frac{r^2}{2}}}%
        \leq%
        1 -\frac{r^2}{4}%
        \eqlab{gogi}
    \end{align}
    since $\ell =1$ and $r \leq 1$.
    
    \begin{figure*}[h]
        \centerline{\includegraphics[page=1]{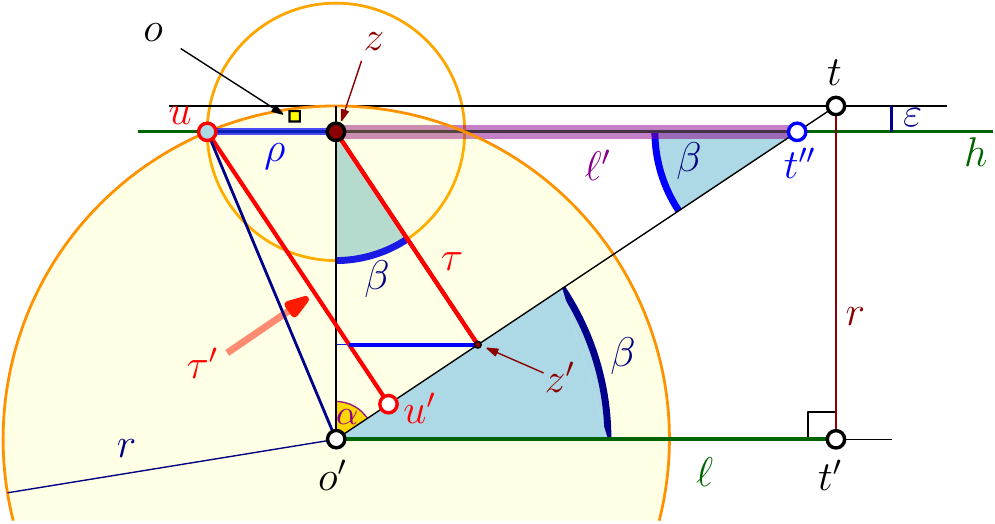}}%
        \caption{Note, that $\opnt$ is not necessarily in the two
           dimensional plane depicted by the figure. All other points
           are in this plane.}%
        \figlab{sanity}%
    \end{figure*}%
    
    \paragraph*{Sanity condition:} %
    Consider the line which is the intersection of the hyperplane
    $\hplane$ and the two dimensional plane spanned by $\pntA, \pntA'$
    and $\opnt'$ (this line is denoted by $\hplane$ in the
    figures). Let $\pntD$ be the point in distance $\rho$ on this line
    from $\pntC$, on the side further away from $\pntA$. Let $\pntA''$
    be the intersection of $\hplane$ with $\pntA \opnt'$.  Next, let
    $\pntD'$ be the nearest point to $\pntD$ on the segment
    $\pntA \opnt'$, see \figref{sanity}. 
    
    We want to argue the distance between $\opnt$ and
    $\CHX{\PntSetA_i}$, can be bounded in terms of the distance
    between $\pntD$ and $\pntD'$, however to do so we need to
    guarantee that $\pntD'$ is in the interior of this segment
    $\pntA \opnt'$.  Setting $\ell' = \distY{\pntC}{\pntA''}$, this
    happens if
    \begin{align*}
      \distY{\pntD'}{ \pntA''} < \distY{\pntA''}{ \opnt'}
      &\iff%
        \distY{\pntD'}{ \pntA''}%
        =%
        \pth{\rho + \ell' }\cos \beta%
        =%
        \pth{\rho + \ell' }%
        \frac{\ell'}{ \distY{\pntA''}{ \opnt'}}%
        <%
        \distY{\pntA''}{ \opnt'}%
      \\
      &\iff%
        \pth{\rho + \ell' }%
        {\ell'} <%
        \distY{\pntA''}{ \opnt'}^2%
        =%
        \pth{\ell'}^2 + \pth{r-\eps}^2.
    \end{align*}
    Thus, we have to prove that
    \begin{math}
        \rho \ell' <%
        \pth{r-\eps}^2.
    \end{math}
    As $\ell' < \ell =1$, we have that this is implied if
    $\rho \leq \sqrt{2 r \eps} < \pth{r-\eps}^2$, and this inequality
    holds if $r \geq 8 \eps^{1/3}$.

    \InNotJVer{%
    \parpic[r]{%
       \includegraphics[page=1]{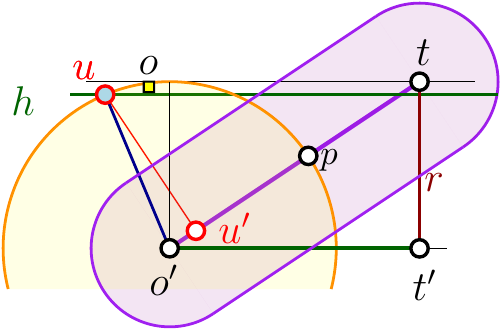}%
    }%
    }
    \InJVer{%
       \begin{figure}[h]
           \includegraphics[page=1]{figs/distances_4}%
       \end{figure}
    }
    
    \medskip%
    \noindent\textbf{Back to the proof:}
    We next bound the distance of $\opnt$ from $\CHX{\PntSetA_i}$.
    Observe that by rotating $\opnt$ around the line $\opnt' \pntA$ we
    can assume that $\opnt$ lies on the plane spanned by
    $\pntA, \pntA', \opnt'$ and its distance to the segment
    $\pntA \opnt'$ has not changed. Now, the set of points in distance
    $r'$ from the segment $\opnt' \pntA$ is a hippodrome, and this
    hippodrome covers a connected portion of $\ballY{\opnt'}{r}$. For
    $r' = \distY{\pntD}{\pntD'}$, by the above sanity condition, this
    hippodrome covers all the points of $\ballY{\opnt'}{r}$ that are
    above $\hplane$.  This implies that $\opnt$ maximizes its distance
    to $\CHX{\PntSetA_i}$ if $\opnt = \pntD$.

    So, let $\idistA = \distY{\pntD}{\pntD'}$.  By the above sanity
    condition the segment $\pntA \opnt'$ and $\pntD \pntD'$ meet at a
    right angle, and hence by similar triangles (see \figref{sanity}),
    we have
    \begin{align*}
      \idistA%
      &=%
        \frac{\ell' + \rho}{\ell'} \idist%
        =%
        \idist + \rho \frac{ \idist}{\ell'}%
        =%
        \idist + \rho \sin \beta%
        =%
        \idist + \rho \frac{r} { \sqrt{\ell^2 +r^2}}%
        =%
        \idist + \rho \frac{r} {\sqrt{1 +r^2}}%
        \leq%
        \idist + \rho r.
    \end{align*}
    This implies, by \Eqref{gogi}, that
    \begin{align*}
        \frac{\distSet{\opnt}{\CHX{\PntSetA_i}}}%
        {\distSet{\opnt}{\CHX{\PntSetA_{i-1}}}}%
        &\leq%
        \frac{\distY{\pntD}{\pntD'}}{\distY{\pntC}{\opnt'}}%
        =%
        \frac{\idistA}{r - \eps}%
        \leq%
        \frac{\idist}{r - \eps} + \frac{\rho r}{r - \eps}%
        \leq%
        \frac{\idist}{r - \eps} + 2\rho%
        \leq %
        1 - {\frac{ r^2}{4}} + 2 \sqrt{r \eps}%
        \leq%
        1 - \eps^{2/3},
    \end{align*}
    if $r \geq 8 \eps^{1/3}$.
\end{proof}

\begin{lemma}
    \lemlab{main}%
    Let $\PntSet$ be a set of $n$ points in $\Re^d$ with diameter
    $\diam = \diamX{\PntSet}$, and let $\eps > 0$ be a parameter, then
    one can compute a set $\PntSetA \subseteq \PntSet$, such
    that
    \begin{compactenum}[(i)]
        \item
        $\distH{\bigl.\CHX{\PntSetA}}{\CHX{\PntSet}} \leq
        \pth{\reachVal + \eps} \diam$,
        \item $m = \cardin{\PntSetA} \leq O\pth{\kopt/\eps^{2/3}}$,
        where $\kopt = \koptY{\PntSet}{\eps}$, and
        \item the running time is $O(nm^2d/\eps^{2})$.
    \end{compactenum}
\end{lemma}

\begin{proof}
    Recall that in any round before the algorithm terminates
    $r_i > \reach\diam = \reachVal \diam$.  Let
    $\OptSet = \OptAprxY{\PntSet}{\eps}$ be any optimal approximating
    set of size $\kopt$. In the $i$\th iteration of the algorithm, for
    some point $\opnt_i\in \OptSet$, its distance to the convex hull
    of $\CHX{\PntSetA_i}$ shrinks by a factor of $1-\eps^{2/3}$, by
    \lemref{shrink}. Conceptually, we charge round $i$ to $\opnt_i$.
    Now, note that by \clmref{optBounded},
    \begin{math}
        \distSet{\opnt_i}{\CHX{\PntSetA_{i-1}}}\geq r_i-\eps\diam >
        (\reach -\eps)\diam \geq \diam\reach/2.
    \end{math}
    Therefore, once the distance of an optimal point $\opnt$ to
    $\CHX{\PntSetA_i}$ falls below
    $\diam \reach /2 = \reachVal \diam/2$, it cannot be charged again
    in any future iteration.  The initial distance of $\opnt$ to
    $\CHX{\PntSetA_0}$ is at most $\diam$. As such, by
    \lemref{shrink}, an optimal point $\opnt$ can get charged at most
    $k$ times, where $k$ is the smallest positive integer such that
    $\pth{1-\eps^{2/3}}^k \diam \leq \reachValHalf \diam$, which holds
    if $\exp\pth{-k\eps^{2/3}} \leq \reachValHalf $. Namely,
    \begin{math}
        k= O\pth{\eps^{-2/3} \log 1/\eps}.
    \end{math}

    Using the same idea of decreasing values of $\eps$, as done in
    \lemref{n:n}, one can improve this bound to
    $O\pth{1/\eps^{2/3}}$. We omit the easy but tedious details.  We
    conclude that the number of iterations performed by the algorithm
    is at most $m = O\pth{ \kopt/\eps^{2/3} }$.

    So the distance of all the points of $\OptSet$ from
    $\CHX{\PntSetA_{m}}$ is at most $\reach \diam$. Now, consider any
    point
    \begin{math}
        \pnt \in \CHX{\PntSet}.
    \end{math}
    Let $\pntA = \nnY{\pnt}{\CHX{\OptSet}}$, and observe that
    $\distY{\pnt}{\pntA} \leq \eps \diam$.  Since
    $\pntA \in \CHX{\OptSet}$, we have that $\pntA$ can be written as
    a convex combination $\pntA= \sum_{i=1}^\nu \alpha_i \opnt_i$,
    where $\alpha_1, \ldots, \alpha_\nu \geq 0$, $\sum_i \alpha_i =1$,
    and $\opnt_1,\ldots, \opnt_\nu \in \OptSet$. For
    $i=1, \ldots, \nu$, let
    $\opnt_i' = \nnY{\opnt_i}{\CHX{\PntSetA_m}}$, and note that
    \begin{math}
        \pntA'= \sum_i \alpha_i \opnt_i' \in \CHX{\PntSetA_m}.
    \end{math}
    Now observe that for all $i$,
    $\distY{\opnt_i}{\opnt_i'} \leq \reach \diam$. In particular,
    $(\opnt_i - \opnt_i')\in \ballY{0}{\reach \diam}$, and hence
    \begin{math}
        \sum_i \alpha_i \pth{\opnt_i - \opnt_i'} \in \ballY{0}{\reach
           \diam}.
    \end{math}
    Therefore
    \begin{math}
        \distSet{\bigl.\pnt}{\CHX{\PntSetA_m}}%
        \leq%
        \distY{\pnt}{\pntA'}%
        \leq%
        \distY{\pnt}{\pntA} + \distY{\pntA}{\pntA'}%
        \leq%
        \eps \diam + \norm{\sum_i \alpha_i \pth{\opnt_i - \opnt_i'}}
        \leq%
        (\eps + \reach) \diam.
    \end{math}
    We conclude that
    \begin{math}
        \distH{\Bigl.\CHX{\PntSetA_m}}{\CHX{\PntSet}} \leq (\eps +
        \reach) \diam.
    \end{math}

    As for the running time, at each iteration, the algorithm computes
    the point in $\PntSet$ furthest away from $\CHX{\PntSetA_i}$.  The
    analysis above assumes these queries are done exactly, which is
    expensive.  
    However, by \lemref{n:n} one can use faster
    $\eps\diam$-approximate queries.  Specifically, in each iteration, 
    for each point $\pnt\in \PntSet$ use \lemref{n:n} to compute an 
    additive $\eps\diam$-approximation to its distance to $\CHX{\PntSetA_i}$, 
    and then select the point in $\PntSet$ with the largest returned approximate distance.
    It is easy to verify this does
    not change the correctness of the algorithm.  Specifically, the
    point $\pntA_i$ chosen in the $i$\th round, may now be $\eps\diam$
    closer to the current convex hull than the furthest point, and so
    in the analysis of \lemref{shrink}, $\opnt_i$ may lie as much as
    $\eps\diam$ above $\pntA_i$.  In particular, the length of
    $\idist$ does not change, however now $\rho$ is only bounded by
    $2\sqrt{r \eps}$ instead of $\sqrt{2r_i\eps}$, and this constant
    factor difference only slightly degrades the constant in front of
    $\eps^{2/3}$ in the lemma statement.  The other effect is that
    when the algorithm stops the distance to the convex hull is
    bounded by $\pth{\reachVal+\eps}\diam$, and this is accounted for
    in the above theorem statement.
    
    Now using \lemref{n:n} directly, it takes
    $O(nmd\allowbreak/\eps^{2})$ time per round to find the
    $\eps\diam$ approximate furthest point, and therefore the total
    running time is $O(nm^2d/\eps^{2})$.
\end{proof}

\subsubsection{Improving the running time further}

The running time of the algorithm of \lemref{main} can be improved
further, but it requires some care.  Let
$\sspace_{i-1} = \vspan(\PntSetA_{i-1})$ denote the linear subspace
spanned by the point set $\PntSetA_{i-1}$, with the orthonormal basis
$v_1, \ldots v_{i-1}$.  For any point $\pnt\in \PntSet$, let
$\pnt_{i-1}'$ denote its orthogonal projection onto the subspace; that
is,
\begin{math}
    \pnt_{i-1}' = \nnY{\pnt}{\sspace_{i-1}} = \sum_{j=1}^i
    \DotProd{\pnt}{v_j} v_j,
\end{math}
and let
$\dSubSpace{i-1}{\pnt} = \distY{\pnt}{\pnt_{i-1}} =
\distSet{\pnt}{\sspace_i}$.
Observe, that for any point $\pntA \in \sspace_{i-1}$ and any point
$\pnt \in \Re^d$, we have that
$\distY{\pnt}{\pntA} = \sqrt{ \distY{\pnt}{\pnt_{i-1}'}^2 +
   \distY{\pnt_{i-1}'}{\pntA}^2 }$
by the Pythagorean theorem, where $\pnt_{i-1}'$ is the projection of
$\pnt$ to $\sspace_{i-1}$.

As such, for any point $\pnt\in \PntSet$, in the beginning of the
$i$\th iteration, the algorithm has the projection and distance of
$\pnt$ to $\sspace_{i-1}$; that is,
\begin{math}
    \pnt_{i-1}' = \pth{\DotProd{\pnt}{v_1}, \ldots,
       \DotProd{\pnt}{v_{i-1}} \bigr. }.
\end{math}
and $\dSubSpace{i-1}{\pnt}$. The algorithm also initially computes for
each point $\pnt \in \PntSet$ its norm $\norm{\pnt}^2$. Therefore, given
any point $\pntA \in \sspace_{i-1}$, its distance to a point
$\pnt \in \PntSet$ can be computed in $O(i)$ time (instead of $O(d)$).
The algorithm also maintains, for every point $\pnt \in \PntSet$, an
approximate nearest neighbor
$\annY{i-1}{\pnt} \in \CHX{\PntSetA_{i-1}}$; that is,
\begin{align*}
    \distSet{\bigl.\pnt}{\CHX{\PntSetA_{i-1}}}%
    \leq%
    \distY{\pnt}{\annY{i-1}{\pnt}}%
    \leq%
    \distSet{\bigl.\pnt}{\CHX{\PntSetA_{i-1}}} + \eps \diam,
\end{align*}
where $\diam=\diamX{\PntSet}$. Naturally, the algorithm also maintains
the distance $\dCHY{i-1}{\pnt} = \distY{\pnt}{\annY{i-1}{\pnt}}$.

Now, the algorithm does the following in the $i$\th iteration:
\smallskip%
\begin{compactenum}[\quad(A)]
    \item Computes, in $O(n)$ time, the point $\pnt \in \PntSet$ that
    maximizes $\dCHY{i-1}{\pnt}$.

    \item Let $\pnt_{i-1}'$ be the projection of $\pnt$ to
    $\sspace_{i-1}$. Computes, in $O(d)$ time, the new vector for the
    basis of $\sspace_i$; that is
    \begin{math}
        \vect_i = \pth{\pnt - \pnt_{i-1}'} / \distY{\pnt}{ \pnt_{i-1}'}.
    \end{math}
    Now $\vect_1, \ldots, \vect_i$ is an orthonormal basis of the
    linear space $\sspace_i$.

    \item For every point $\pnt \in \PntSet$, update its projection
    $\pnt_{i-1}'$ into $\sspace_{i-1}$ into the projection of $\pnt$
    into $\sspace_{i}$, by computing $\DotProd{\pnt}{v_i}$. Also,
    update
    $\dSubSpace{i}{\pnt} = \sqrt{\dSubSpace{i-1}{\pnt}^2
          -\DotProd{\pnt}{v_i}^2}$.%

    \item %
    \itemlab{update:n:n}%
    Let $\PntSet'$ denote the projected points of $\PntSet$ into
    $\sspace_i$.  For every $\pnt \in \PntSet$, we need to update
    $\annY{i-1}{\pnt}$ to $\annY{i}{\pnt}$ (and the associated
    distance). To this end, the algorithm of \lemref{n:n} is called on
    $\pnt_i'$ and $\PntSetA_i$ (all lying in the subspace $\sspace_i$
    which is of dimension $i$). Importantly, the algorithm of
    \lemref{n:n} is being warm-started with the point
    $\annY{i-1}{\pnt}$. Let $\#_i(\pnt)$ be the number of iterations
    performed inside the algorithm of \lemref{n:n} to update the
    nearest-neighbor to $\pnt$. Observe, that the running time for
    $\pnt$ is $O\pth{ \#_i(\pnt) i^2 }$, since
    $i = \cardin{\PntSetA_i}$, the points lie in an $i$ dimensional
    space, and as such, every iteration of the algorithm of
    \lemref{n:n} takes $O(i^2)$ time.
\end{compactenum}%

\begin{lemma}%
    For $m = O \bigl( \kopt / \eps^{2/3} \bigr)$, the running time of
    the above algorithm is
    $O\bigl( nm \bigl( d + m/\eps^2 + m^2 \bigr) \bigr) $.
\end{lemma}

\begin{proof}
    The algorithm performs $m = O\pth{ \kopt/\eps^{2/3}}$ iterations,
    and this bound the dimension of the output subspace. Every
    iteration of the algorithm takes $O(n d)$ time, except for the
    last portion of updating the approximate nearest point for all the
    points of $\PntSet$ (i.e., \itemref{update:n:n}). The key
    observation is that $\sum_i \pth{\#_i(\pnt) - 1} = O(1/\eps^2)$,
    since if the algorithm of \lemref{n:n} runs
    $\alpha = \#_i(\pnt) > 1$ iterations, then the distance of $\pnt$
    to the convex-hull shrinks by a factor of
    $(1-\eps^2/2)^\alpha$. Arguing as in the proof of \lemref{n:n},
    this can happen $O(1/\eps^2)$ times before $\pnt$ is in distance
    at most $\eps \diam$ from the convex-hull, and can no longer be
    updated. As such, for a single point $\pnt \in \PntSet$, the
    operations in \itemref{update:n:n} takes overall
    \begin{math}
        \sum_{i=1}^m O\pth{ i^2 \pth{\#_i(\pnt) - 1} }%
        =%
        O\pth{ m^2 (m+1/\eps^2)}
    \end{math}
    time.  This implies the overall running time of the algorithm is
    $O\pth{ n \pth{dm + m^2/\eps^2 + m^3 }}$.
\end{proof}%

\subsubsection{The result}

\begin{theorem}
    \thmlab{main}%
    Let $\PntSet$ be a set of $n$ points in $\Re^d$ with diameter
    $\diam = \diamX{\PntSet}$, and let $\eps > 0$ be a parameter, then
    one can compute a set $\PntSetA \subseteq \PntSet$, such
    that
    \begin{compactenum}[\qquad(i)]
        \item
        $\distH{\bigl.\CHX{\PntSetA}}{\CHX{\PntSet}} \leq
        \pth{\reachVal + \eps} \diam$, and
    
        \item $\cardin{\PntSetA} \leq O\pth{\kopt/\eps^{2/3}}$, where
        $\kopt = \koptY{\PntSet}{\eps}$.
    \end{compactenum}
    \smallskip%
    The running time of the algorithm is
    \begin{math}
        O\pth{ nm \pth{ d + m/\eps^2 + m^2 }}.
    \end{math}
    for $m = O \bigl( \kopt / \eps^{2/3} \bigr)$.  (Here, the
    constants hidden in the $O$ are independent of the dimension.)
\end{theorem}

\begin{remark:unnumbered}
    (A) The constants hidden in the $O$ notation used of \thmref{main}
    are independent of the dimension. In comparison to the other
    algorithms in this paper, the approximation quality is slightly
    worse. However, the advantage is a drastic improvement in the size
    of the approximation.

    (B) The running time of the algorithm of \thmref{main} can be
    further improved, by keeping track for each point
    $\pnt \in \PntSet$, and each point $\pntA \in \PntSetA_i$, the
    distance of $\pntA$ from the hyperplane (in $\sspace_i$) that
    determines whether or not the approximate nearest neighbor to
    $\pnt$ needs to be recomputed. By careful implementation, this can
    be done in the $i$\th iteration in $O(i n)$ time (updating
    $O(i n)$ such numbers in this iteration).  This improves the
    running time to $O\pth{ nm \pth{ d + m/\eps^2 }} $. Motivated by
    our laziness we omit the messy details.
\end{remark:unnumbered}

\begin{remark:unnumbered}
    Note that the algorithm is a simple iterative process, which is
    oblivious to the value of the diameter $\diam = \diamX{\PntSet}$
    and does not use it directly anywhere.  Nevertheless, after
    $O\pth{\kopt/\eps^{2/3}}$ iterations the solution is an
    ${ \pth{\reachVal + \eps} \diam}$-approximation to the convex
    hull.  In practice, one may not know the value of $\kopt$, and so
    this value cannot be used in a stopping condition.  However, it is
    easy to get a $2$-approximation $\diam'$, such that
    $\diam \leq \diam' \leq 2\diam$, by a linear scan of the points.
    Then, one can use the check
    \begin{math}
        \distSet{\bigl.\pntA_i}{\CHX{\PntSetA_i}}%
        =%
        \distH{\CHX{\PntSet}}{\bigl.\CHX{\PntSetA_i}}%
        \leq%
        \pth{\reachVal + \eps} \diam'/2
    \end{math}
    as a stopping condition, where $\PntSetA_i$ is the current
    approximation.
\end{remark:unnumbered}






 \providecommand{\CNFX}[1]{ {\em{\textrm{(#1)}}}}
  \providecommand{\tildegen}{{\protect\raisebox{-0.1cm}{\symbol{'176}\hspace{-0.03cm}}}}
  \providecommand{\SarielWWWPapersAddr}{http://sarielhp.org/p/}
  \providecommand{\SarielWWWPapers}{http://sarielhp.org/p/}
  \providecommand{\urlSarielPaper}[1]{\href{\SarielWWWPapersAddr/#1}{\SarielWWWPapers{}/#1}}
  \providecommand{\Badoiu}{B\u{a}doiu}
  \providecommand{\Barany}{B{\'a}r{\'a}ny}
  \providecommand{\Bronimman}{Br{\"o}nnimann}  \providecommand{\Erdos}{Erd{\H
  o}s}  \providecommand{\Gartner}{G{\"a}rtner}
  \providecommand{\Matousek}{Matou{\v s}ek}
  \providecommand{\Merigot}{M{\'{}e}rigot}
  \providecommand{\CNFSoCG}{\CNFX{SoCG}}
  \providecommand{\CNFCCCG}{\CNFX{CCCG}}
  \providecommand{\CNFFOCS}{\CNFX{FOCS}}
  \providecommand{\CNFSODA}{\CNFX{SODA}}
  \providecommand{\CNFSTOC}{\CNFX{STOC}}
  \providecommand{\CNFBROADNETS}{\CNFX{BROADNETS}}
  \providecommand{\CNFESA}{\CNFX{ESA}}
  \providecommand{\CNFFSTTCS}{\CNFX{FSTTCS}}
  \providecommand{\CNFIJCAI}{\CNFX{IJCAI}}
  \providecommand{\CNFINFOCOM}{\CNFX{INFOCOM}}
  \providecommand{\CNFIPCO}{\CNFX{IPCO}}
  \providecommand{\CNFISAAC}{\CNFX{ISAAC}}
  \providecommand{\CNFLICS}{\CNFX{LICS}}
  \providecommand{\CNFPODS}{\CNFX{PODS}}
  \providecommand{\CNFSWAT}{\CNFX{SWAT}}
  \providecommand{\CNFWADS}{\CNFX{WADS}}

 \providecommand{\CNFX}[1]{ {\em{\textrm{(#1)}}}}
  \providecommand{\tildegen}{{\protect\raisebox{-0.1cm}{\symbol{'176}\hspace{-0.03cm}}}}
  \providecommand{\SarielWWWPapersAddr}{http://sarielhp.org/p/}
  \providecommand{\SarielWWWPapers}{http://sarielhp.org/p/}
  \providecommand{\urlSarielPaper}[1]{\href{\SarielWWWPapersAddr/#1}{\SarielWWWPapers{}/#1}}
  \providecommand{\Badoiu}{B\u{a}doiu}
  \providecommand{\Barany}{B{\'a}r{\'a}ny}
  \providecommand{\Bronimman}{Br{\"o}nnimann}  \providecommand{\Erdos}{Erd{\H
  o}s}  \providecommand{\Gartner}{G{\"a}rtner}
  \providecommand{\Matousek}{Matou{\v s}ek}
  \providecommand{\Merigot}{M{\'{}e}rigot}
  \providecommand{\CNFSoCG}{\CNFX{SoCG}}
  \providecommand{\CNFCCCG}{\CNFX{CCCG}}
  \providecommand{\CNFFOCS}{\CNFX{FOCS}}
  \providecommand{\CNFSODA}{\CNFX{SODA}}
  \providecommand{\CNFSTOC}{\CNFX{STOC}}
  \providecommand{\CNFBROADNETS}{\CNFX{BROADNETS}}
  \providecommand{\CNFESA}{\CNFX{ESA}}
  \providecommand{\CNFFSTTCS}{\CNFX{FSTTCS}}
  \providecommand{\CNFIJCAI}{\CNFX{IJCAI}}
  \providecommand{\CNFINFOCOM}{\CNFX{INFOCOM}}
  \providecommand{\CNFIPCO}{\CNFX{IPCO}}
  \providecommand{\CNFISAAC}{\CNFX{ISAAC}}
  \providecommand{\CNFLICS}{\CNFX{LICS}}
  \providecommand{\CNFPODS}{\CNFX{PODS}}
  \providecommand{\CNFSWAT}{\CNFX{SWAT}}
  \providecommand{\CNFWADS}{\CNFX{WADS}}


\end{document}